\documentclass[12pt,a4paper]{article}

\usepackage{amscd,amsfonts,amsmath,amssymb,amsthm,bbm,bm,latexsym,mathrsfs}
\usepackage{epsfig,graphics,graphicx,natbib,longtable}
\usepackage[english]{babel}
\usepackage[usenames]{color}
\usepackage{bookmark}
\usepackage{booktabs}
\usepackage{sgame}
\usepackage[top=1.7in, bottom=1.7in, left=0.9in, right=0.9in]{geometry}
\allowdisplaybreaks
\makeatother

\usepackage{amsthm}
\usepackage{color}
\usepackage[latin1]{inputenc}
\usepackage{amsfonts}
\usepackage{amsthm}
\usepackage{amssymb,bm}
\usepackage{amsmath,enumerate,rotate}
\usepackage{amssymb,latexsym,mathrsfs}
\usepackage{algorithm}
\usepackage{epsfig,setspace}
\usepackage{graphicx}
\usepackage[english]{babel}
\usepackage{natbib}
\usepackage{color}
\usepackage{multirow}
\usepackage{graphicx}
\usepackage{subcaption}

\usepackage{verbatim}

\newcommand{\ddr}{\mathrm{d}}
\newcommand{\edr}{\mathrm{e}}

\newtheorem{thm}{\noindent Theorem}

\newcommand{\E}{\mathbb{E}}

\renewcommand{\P}{\mathbb{P}}

\newcommand{\R}{\mathbb{R}}

\newcommand{\Fcr}{\mathcal{F}}

\def\eqd{\stackrel{\mbox{\scriptsize{d}}}{=}}

\theoremstyle{remark}

\theoremstyle{plain}

\begin{document}

\title{\vspace{-50pt} Limiting behaviour of the stationary search cost distribution driven by a generalized gamma process}

\author{
Alfred Kume\textsuperscript{a} \hspace{15pt}
Fabrizio Leisen\textsuperscript{a}\thanks{Corresponding Author email address:  \href{mailto:fabrizio.leisen@gmail.com}{fabrizio.leisen@gmail.com} (F. Leisen).} \hspace{15pt}
 Antonio Lijoi\textsuperscript{b}\thanks{Also affiliated to the Bocconi Institute of Data Science and Analytics, Italy.}  
        \\ 
        \vspace{5pt}
        \\
        {\centering {\small \textsuperscript{a} 
        University of Kent, UK \hspace{18pt}
        \textsuperscript{b}
             Bocconi University, Italy}} \vspace{5pt} \\
     }

\date{}
\maketitle

\abstract{Consider a list of 
labeled objects 
that are organized in a heap. At each time, object $j$ is selected with probability $p_j$ and moved to the top of the heap. This procedure defines a
Markov chain on the set of permutations which is referred to in the literature 
as Move-to-Front rule. The present contribution focuses on the stationary search cost, namely the position of the requested item in the heap when the Markov chain is in equilibrium. We consider the scenario where the number of objects is infinite and the probabilities $p_j$'s are 
defined as the normalization of the increments of a subordinator. In this setting, we provide an exact formula for the moments of any order of the stationary search cost distribution. We illustrate the new findings 
in the case of a generalized gamma subordinator and deal with an extension to the two--parameter Poisson--Dirichlet process, also known as Pitman--Yor process.  \\
\bigskip

}

\section{Introduction}\label{sec:intro}

The Move-to-Front (MtF) rule, also known as Tsetlin library or Heaps process, identifies a well--known stochastic process, which arises in various applied research areas. It is used to describe an experiment whereby objects are requested at a random instant with a certain probability from a finite set of items in a serial list; when an object is requested, it is moved to the front of the list and the positions of the remaining items are unchanged. This procedure defines an underlying Markov chain on the set of permutations and an important goal that has been pursued in this area is the determination of the  probability distribution of the search cost, which is defined as the depth of the requested item in the list. The investigation of the MtF rule has attracted considerable interest of several authors working at the interface of Computer Science and Probability. For example, \cite{fill96a,fill96b} and \cite{fh} are important reference works in the area and investigated the model under different assumptions 
for the request probabilities. A noteworthy extension occurs when the number of items is infinite. 
In this context,  \cite{bp,bhp} and \cite{bhp2} studied the behaviour of the search cost distribution for random request probabilities defined by normalized positive independent and identically distributed (i.i.d.) random variables. \cite{DoubleHattori} illustrate a link between sales ranks of online shops and the MtF rule. \cite{Jele99} considered different scaling limits for the search cost when suitable optimality conditions in the list are  specified. In \cite{bf} the behaviour of the search cost distribution for an infinite number of objects is studied under the assumption that a law of large numbers is satisfied on the unnormalized weights that are used to identify the probability masses with which items are requested; additionally, some connections with results of \cite{Jele08} are displayed. Finally, it is worth remarking that the stationary  distribution of the MtF Markov chain has been employed for modeling partial ranking data drawn from a finite collection of items. Such a model has been named after R.L. Plackett and R.D. Luce [see \cite{Plackett} and \cite{luce}]. A nice recent application, within a Bayesian nonparametric framework, has been proposed in \cite{Caron14}. See also references therein. 

The focus of the present paper will be on the determination of the stationary search cost distribution when the number of objects goes to infinity and the random request probabilities arise from the normalization of the increments of a subordinator. We will show that the limiting Laplace transform of the search cost distribution can be expressed in terms of the Laplace exponent of the underlying subordinator. As a by--product of this main result, we will be able to evaluate the moments of any order for the stationary search cost distribution. We will test our findings on the normalized generalized gamma process. Not only this includes as special cases other noteworthy instances of random discrete probability measures, such as the Dirichlet process and the normalized $\gamma$--stable subordinator, but it has also been extensively applied in the Bayesian nonparametric inference literature. See, e.g., \cite{lmp}. We will also provide a pointer to the two--parameter Poisson--Dirichlet process, also known as Pitman--Yor process (\cite{PitmanYor}), thanks to its representation as a mixture of generalized gamma processes with a base measure having random total mass.

\subsubsection*{Outline of the paper}

In Section 2 we will provide some preliminary notions about the stationary search cost distribution and recall a few relevant results which will be used in the paper. 
In Section 3 we will provide an expression for the limiting Laplace transform when an infinite--activity subordinator is used to model the request probabilities. As a consequence, we will be able to derive a general expression for the moments of any order of the search cost distribution. In Section 4 we will specify the moment formula when the request probabilities arise from normalization of a generalized gamma process or a two--parameter Poisson--Dirichlet process. As a byproduct, we will be able to recover results in \cite{kingman} and \cite{llp} as special cases. 

\section{The stationary search cost distribution}

In this Section we recall some preliminary definitions and results about the stationary search cost distribution. This part is also helpful since it sets up some notation that will be used henceforth. In order to describe the experiment that gives rise to the search cost distribution, 
consider a collection of items $I_1,I_2,\dots,I_n$ that are organized in a heap. For instance, books on a library shelf, files stored in a computer, etc. Suppose that the probability of requesting item $I_j$ is $p_j$, $j=1,\dots,n$. At each time, if 
$I_j$ is selected 
it is placed at the top of the heap. Successive requests are independent and, at each time, only one item may be removed from the heap. The underlying stochastic process is a Markov chain on permutations of the elements of the list and it is known as the \textit{Move-to-Front} rule. See, e.g., \cite{donnelly}.  The stationary distribution of this Markov chain is  
\[
P(\bm{\sigma}=\sigma)=p_{\sigma_1}\frac{p_{\sigma_2}}{1-p_{\sigma_1}}\frac{p_{\sigma_3}}{1-(p_{\sigma_1}+p_{\sigma_2})}\cdots\frac{p_{\sigma_n}}{1-(p_{\sigma_1}+\cdots+p_{\sigma_{n-1}})}
\]
where $\bm{\sigma}=(\sigma_1,\ldots,\sigma_n)$ is  a random permutation of $(1,\ldots,n)$. In general, it is assumed that the chain starts deterministically in permutation.  

The search cost is the position of the requested item in the heap or, equivalently, the number of items to be removed from the heap in order to find the requested one. In this setting, it might be of interest to determine the distribution of the search cost when the underlying
Markov chain is at equilibrium. In order to do so, we assume that the probabilities $p_1,\dots, p_n$ are random. 
In particular, if $w_1,\dots,w_n$ is a sequence of independent random variables, one can define
$$p_i=\frac{w_i}{\sum_{i=1}^n w_i},\quad i=1,\dots, n.$$
In \cite{kingman} the $p_i$'s are expressed as the normalized increments of a stochastic process
\begin{equation}
w_i=\xi_{t_i}-\xi_{t_{i-1}},\qquad i=1,\ldots,n
\label{eq:increments_subordinator}
\end{equation}
where $0=t_0<t_1<\cdots<t_n=1$ and $\xi=\{\xi_t$: $t\in[0,1]\}$ is a subordinator, i.e. a process with independent increments, almost surely increasing paths and such that $\P[\xi_0=0]=1$. In particular, it is assumed that $\xi$ is a \textit{Gamma} subordinator, i.e.
\begin{equation}\label{eq:lapl_w}
\phi_{i,n}(s)=\mathbb{E}\left[e^{-sw_i}\right]=
(1+s)^{-(t_i-t_{i-1})}\quad s\geq 0,
\end{equation}
or a \textit{$\gamma$-stable} subordinator for some $\gamma\in(0,1)$, i.e.
\begin{equation}\label{eq:lapl_w}
\phi_{i,n}(s)=\mathbb{E}\left[e^{-sw_i}\right]=e^{-(t_i-t_{i-1})s^{\gamma}}\quad s\geq 0.
\end{equation}
In both cases, an expression of the expected stationary search cost has been determined in \cite{kingman}, when the number of items $n$ is taken to diverge to $\infty$ in such a way that $\max_{1\le i\le n} (t_i-t_{i-1})\to 0$. Such a result has been extended by \cite{lp} to any subordinator and to the case where the list contains a finite number of items. In \cite{llp} the limiting behaviour of the moments of any order of the stationary search cost distribution is investigated with a $\gamma$--stable subordinator defining the weights $w_i$ as in \eqref{eq:increments_subordinator}. 
Usually, when move-to-front processes are considered, the moves are done at each discrete unit of time. 
Nonetheless, for illustrative purposes it may be convenient to consider a continuous time specification of the process. 
For instance, \cite{fh} considered the case where the moves are done at the time points of a Poisson process of intensity $1$ on $[0,\infty)$. This yields a \textit{continuized} Markov chain, which has the same stationary distributions as the one arising in the discrete move--to--the--front case.  
In this setting, let
$$\mathbb{I}_{ji}(t)=\left\lbrace\begin{array}{ll}
1&\mbox{ if item } I_j \mbox{ precedes item } I_i \mbox{ in the list at time } t\\
0&\mbox{ otherwise}\\
\end{array}
\right.$$
The \textit{search cost} 
for item $I_i$ at time $t$ is defined as
$$S_{n,i}(t)=\sum_{j\neq i}\mathbb{I}_{ji}(t).$$ 
Letting $R$ denote a random variable independent of the MtF Markov chain that identifies the label of the selected item, namely $P(R=i)=p_i$, then \textit{search cost} at time $t$ equals 
$S_n(t)=S_{n,R}(t).$
The random variable
$$
S_n\equiv S_n(\infty)=\lim_{t\rightarrow +\infty} S_{n}(t).
$$
is termed the \textit{stationary} search cost. It is worth noting that the stationary search cost can be seen as a size-biased pick from a size-biased permutation of a random discrete distribution on the positive integers minus one. To this end, the reader may refer to \cite{PPY}
for a thorough investigation on size-biased permutations of ranked jumps of a subordinator; see also \cite{LibroPitman}. More recently, \cite{PitmanTran} have investigated a finite dimensional analogue of the discrete random probability measure studied in \cite{PPY}.

In \cite{llp} one can find an explicit expression for the moments of any order $k\ge 1$ of $S_n$, when the request probabilities $p_i$ are the normalized increments of a $\gamma$--stable subordinator, as $n\to\infty$. In particular, one has

\begin{thm} {\rm (\cite{llp})}\label{TeoLLP}
If the $(p_1,\ldots,p_n)$ are determined by normalizing the increments of a $\gamma$-stable subordinator as in \eqref{eq:lapl_w}, with $t_i-t_{i-1}=1/n$ for each $i \in \{1, \ldots, n\}$, then
$$
\lim_{n \rightarrow \infty} \mathbb{E}(S_n^k) = \left\{ \begin{array}{ll}
\sum_{l=1}^k \frac{(l!)^2}{(\frac{1}{\gamma}-l-1)_{l}} a_l^{(k)}&\mbox{ if }\gamma<\frac{1}{k+1}\\ \infty &{\mbox{otherwise}}\\
\end{array}\right.
$$
where $a_l^{(k)}$, $l=1,\dots,k$, are the Stirling numbers of second kind. 
\end{thm}

The main tool for achieving the previous result is the Laplace transform of $S_n$, which can be determined according to the following 

\begin{thm} {\rm (\cite{bp})}
If $\{w_i: i=1,\ldots,n\}$ are non--negative independent random variables and $p_i=w_i/\sum_{j=1}^n w_j$ for each $i=1,\ldots,n$, then
\begin{equation}\label{eqn:lt}
\phi_{S_n}(s) = \sum_{i=1}^n \int_0^\infty \left( \int_t^\infty \phi_{i,n}''(r) \prod_{j \neq i}h_{t,s,j,n}(r) \,\ddr r \right) \,\ddr t \;,
\end{equation}
for all $s \geqslant 0$, where $\phi_{j,n}(s)=\E[\edr^{-s w_j}]$ and 
$$
h_{t,s,j,n}(r) = \phi_{j,n}(r)+\edr^{-s} (\phi_{j,n}(r-t)-\phi_{j,n}(r)) \;, \quad t \geqslant 0, r \geqslant t \;.
$$
\end{thm}
In the next section we will tackle the problem of the determination of the moments by considering a more general setting. 

\section{Subordinators--based search cost distribution}

Suppose $\xi=\{\xi_t:\: t\in [0,1]\}$ is a stochastic process defined on some probability space $(\Omega,\Fcr,\P)$ such that
\begin{itemize}\addtolength{\itemsep}{-0.4\baselineskip}
\item[(i)] for any $n\ge 1$ and $0=t_0<t_1\,\cdots\,<t_n\le 1$, the random variables $\xi_{t_i}-\xi_{t_{i-1}}$ are independent ($i=1,\ldots,n$);
\item[(ii)] $\P[ \xi_0=0]=1$ and $\xi_t-\xi_s\eqd\xi_{t-s}$ for any $1\ge t\ge s\ge 0$;
\item[(iii)] $t\mapsto\xi_t$ is right--continuous and non--decreasing, with $\P$--probability 1;
\end{itemize}
Henceforth $\xi$ is termed \textit{subordinator} and there exists a measure $\nu$ on $\R^+$ such that $\int_0^\infty \min\{1,y\}\,\nu(\ddr y)<\infty$ and 
\begin{equation}
  \label{eq:lapl_crm}
  \psi(s):=-\frac{1}{t}\log\Big(\E[\edr^{-s \xi_t}]\Big)=\int_0^\infty\left[1-\edr^{-sy}\right]\,
\nu(\ddr y)
\end{equation}
for any $s\ge 0$. The measure $\nu$ is often referred to as the \textit{L\'evy measure} of $\xi$, whereas $\psi$ is the so--called \textit{Laplace exponent} of $\xi$. Noteworthy examples the Gamma process, which is identified by
\begin{equation}
\nu(\ddr y)=y^{-1}\:\edr^{-y}\:\ddr y,
\end{equation}
and the $\gamma$--stable process, with $\gamma\in(0,1)$, whose L\'evy measure is
\begin{equation}
  \label{eq:gamma_st_levy}
  \nu(\ddr y)=\frac{\gamma}{\Gamma(1-\gamma)}\: y^{-1-\gamma}\:\ddr y.
\end{equation}
It is obvious that $\psi$ fully characterizes $\xi$. Hence, in order to identify the cost search distribution, with request probabilities obtained as transformations of increments of subordinators, we can target the determination of the Laplace transform as a function of $\psi$. Indeed, we provide a closed form expression for
$$\phi_{S}(s)=\lim_{n\rightarrow\infty} \phi_{S_n}(s)$$ 
and for the moments $\E[S_n^k]$ as $n\to\infty$, for any $k\ge 1$, in terms of the Laplace exponent $\psi$ of the underlying subordinator $\xi$. In order to state the main result of the paper, it is convenient to introduce the following quantity, 
\begin{equation}\label{UI2}
I_n(l)=\int_0^{\infty}\edr^{-\frac{n-2}{n}\psi(r)}\int_0^{\infty}[\psi'(r+t)]^2\edr^{-\frac{2}{n}\psi(r+t)}[\psi(r+t)-\psi(r)]^{l-1}\ddr t\ddr r.
\end{equation}
for every $l=1,\dots,k$.
\begin{thm}\label{TeoCRM}
 If the $(p_1,\ldots,p_n)$ are determined by normalizing the increments of a subordinator in \eqref{eq:increments_subordinator}, with $t_i-t_{i-1}=1/n$ 
for each $i \in \{1, \ldots, n\}$, then
\begin{equation}\label{FirstContri}
\begin{split}
\phi_S(s)
&=-\int_0^{\infty}\int_{0}^{\infty} \psi''(x+y)\edr^{-\psi(x+y)} \edr^{-\edr^{-s}[\psi(x)-\psi(x+y)]}\,\ddr x\,\ddr y.
\end{split}
\end{equation}
For any positive integer $k$ such that  $\sup_{n} I_n(l)<\infty$, $l=1,\dots,k$, one has 
\begin{equation}\label{MomentGen}
\lim_{n \rightarrow \infty} \mathbb{E}[S_n^k] =\sum_{l=1}^k  a_l^{(k)}\Psi(l),
\end{equation}
where $a_l^{(k)}$, $l=1,\dots,k$, are the Stirling numbers of second kind and
\begin{equation}\label{Momenti}
\Psi(l)=-\int_0^{\infty}\int_{0}^{\infty} [\psi(x+y)-\psi(x)]^l \psi''(x+y)\,\edr^{-\psi(x)}\:
\ddr x\, \ddr y.
\end{equation}
\end{thm}

\begin{proof} If $t_i-t_{i-1}=1/n$, for every $i=1,\dots,n$, then one obviously has $\phi_{i,n}(r)=\phi_n(r)=\edr^{-\psi(r)/n}$. If one, now, considers the expression of the Laplace transform in \eqref{eqn:lt}, it is apparent that $h_{t,s,i,n}=h_{t,s,n}$, for any $i=1,\ldots,n$ and 
\[
\phi_{S}(s)=\lim_{n\rightarrow\infty}\phi_{S_n}(s)=\lim_{n\to\infty}
n \,\int_0^\infty \int_{t}^\infty \phi_n ''(r)  h^{n-1}_{t,s,n}(r) \ddr r\, \ddr t
\]
As one trivially has
$$
\phi_n''(r)=\left[\frac{(\psi'(r))^2}{n^2}-\frac{\psi''(r)}{n}\right]\phi_n(r),
$$
then \ $\lim_{n\to\infty} n\phi_n''(r)= -\psi''(r)$. 
On the other hand, 
\[
h^{n-1}_{t,s,n}(r) =
\phi_n(r)^{n-1} \left[ 1+\edr^{-s}\left(\frac{\phi_n(r-t)}{\phi_n(r)}-1\right) \right]^{n-1}
\]
Since $\{\phi_n(r)\}^{n-1} \to \edr^{-\psi(r)}$
and 
$$
\frac{\phi_n(r-t)}{\phi_n(r)}-1=\edr^{-(\psi(r-t)-\psi(r))/n}-1\sim -\frac{\psi(r-t)-\psi(r)}{n}+o\left(\frac{1}{n}\right) $$
as $n\to\infty$, 
we conclude that
$$
\lim_{n\rightarrow\infty} h^{n-1}_{t,s,n}(r)=\edr^{-\psi(r)} \,\edr^{-\edr^{-s}[\psi(r-t)-\psi(r)]}.
$$
Then
\[ 
\phi_S(s)=-\int_0^{\infty}\int_{t}^{\infty} \psi''(r)\edr^{-\psi(r)} \edr^{-\edr^{-s}[\psi(r-t)-\psi(r)]}\:\ddr r\,\ddr t
\] 
and the conclusion follows from the simple change of variable $x=r-t$. So far, we have proved the weak convergence of $S_n$ to $S$. As far as the determination of the moments of the limiting variable $S$ is concerned, note that the Laplace transform $\phi_S(s)$ can be rewritten as 
$$\phi_S(s)=\int_0^{\infty}\int_{0}^{\infty}f(x,y) \edr^{-(1-\edr^{-s})[\psi(x+y)-\psi(x)]}\,\ddr x\,\ddr y,$$
where
$$f(x,y)=-\psi''(x+y)\edr^{-\psi(x)}.$$
Hence, $S$ is equal in distribution to a mixture of Poisson distributions with parameter $\psi(x+y)-\psi(x)$ where $f(x,y)$ is the mixing law. In view of this representation, one can determine the moments of $S$ as a mixture of moments of the underlying Poisson distributions. It is well known that the $k$-th moment of a Poisson distribution with parameter $\lambda$ is $\sum_{l=1}^k a_l^{(k)}\lambda^{l}$. This immediately yields that the $k$-th moment of $S$ coincides with the right hand side of equation \eqref{MomentGen}. According to the Corollary of Theorem 25.12 in \cite{Billy}, in order to establish the equality in equation \eqref{MomentGen}, we need to prove that 
$\sup_n \mathbb{E} [S_n^k]<\infty.$ Following \cite{llp}, one has 
\begin{equation}\label{allmoments}
\mathbb{E}(S_n^k)=\sum_{l=1}^k a_l^{(k)}M_{l,n}(0),
\end{equation}
where 
\begin{align*}
M_{l,n}(0)&=l\cdot n(n-1)\cdots(n-l)
\int_0^{\infty}[\phi_n(r)]^{n-l-1}\int_0^{\infty}\phi'_n(r+t)\,
\phi'_{n}(r+t)\\
&\quad\qquad \times [\phi_{n}(r)-\phi_{n}(r+t))]^{l-1}\ddr t\ddr r.
\end{align*}
The above equation suggests that it is enough to prove that $\sup_{n}M_{l,n}(0)<\infty$, for every $l=1,\dots,k$. Hence, 
\begin{align*}
M_{l,n}(0)&=l\frac{(n-1)\cdots(n-l)}{n}
\int_0^{\infty}\edr^{-\frac{n-l-1}{n}\psi(r)}\int_0^{\infty}[\psi'(r+t)]^2\edr^{-\frac{2}{n}\psi(r+t)}\\
&\quad\qquad \times [\edr^{-\psi(r)/n}-\edr^{-\psi(r+t)/n}]^{l-1}\ddr t\ddr r\\
&=l\frac{(n-1)\cdots(n-l)}{n}
\int_0^{\infty}\edr^{-\frac{n-2}{n}\psi(r)}\int_0^{\infty}[\psi'(r+t)]^2\edr^{-\frac{2}{n}\psi(r+t)}\\
&\quad\qquad \times [1-\edr^{-[\psi(r+t)-\psi(r)]/n}]^{l-1}\ddr t\ddr r
\end{align*}
By using the well known inequality $1-\edr^{-x}\leq x$ we get
\begin{align*}
M_{l,n}(0)
&\leq l\frac{(n-1)\cdots(n-l)}{n^l} I_n(l),
\end{align*}
and hence, the convergence in equation \eqref{MomentGen} follows from the assumption that $\sup_{n} I_n(l)<\infty$. 
\end{proof}

In the next Section we will use the previous result to derive the expression of the limiting moments when the $p_i$'s are obtained by means of a normalized generalized gamma process or of a two--parameter Poisson--Dirichlet process. 

\section{The generalized gamma process and the Two-parameter Poisson-Dirichlet process}

The generalized gamma process has been introduced in \cite{Brix} for constructing shot noise Cox processes. It is characterized by the following L\'evy measure 
\begin{equation}\label{GenGamma}
\nu(dy)=
\Gamma(1-\gamma)^{-1} y^{-(1+\gamma)}e^{-uy}\ddr y,
\end{equation}
where $\gamma\in (0,1)$ and $u\geq 0$. It turns out that the Laplace exponent, evaluated at any $s\ge 0$, is 
$$\psi(s)=\frac{(u+s)^{\gamma}-u^{\gamma}}{\gamma}.$$ 
In this case $\bar{\xi}=\{\xi_t/\xi_1:\: t\in[0,1]\}$ identifies the normalized generalized gamma process and it will be denoted with the notation NGG($\gamma,u$). This random probability measure has been used for density estimation in Bayesian nonparametric mixture models and, when its distribution is the directing measure of a sequence of exchangeable random elements, the associated predictive distributions can be determined in closed form. See \cite{lmp}. It is also worth stressing that it includes as special cases both the normalized $\gamma$--stable ($u=0$) and the Dirichlet processes ($\gamma\to 0$). Furthermore, mixtures of normalized generalized gamma processes induce a two--parameter Poisson--Dirichlet process. Specifically, let $Z$ be a random variable with density function
$$
f_Z(z)=\frac{1}{\Gamma(\theta/\gamma)}\: z^{(\theta/\gamma)-1}\,\edr^{-z},
$$
for any $\theta>0$ and $\gamma\in(0,1)$, i.e. $Z$ is a Gamma random variable with parameters $(\theta/\gamma)$ and $1$. If $\bar{\xi}$ is a NGG($\gamma,1$) independent from $Z$, from Proposition 21 in \cite{PitmanYor}, one has that the normalized process $\xi_{Zt}/\xi_Z$ has the same distribution as a Poisson--Dirichlet process with parameters $(\gamma,\theta)$. Hence, one can define weights
\begin{equation}\label{2Pincre}
w_i=\xi_{Z t_i}-\xi_{Z t_{i-1}}\qquad i=1,\ldots,n
\end{equation}
which, in turn, entail
\[
\phi_{i,z}(s)=\E\Big[\edr^{-s w_i}\,\Big|\,z\Big]=\edr^{-z(t_i-t_{i-1})[(1+s)^\gamma-1]}.
\]
As we will see in the proof of Theorem \ref{MomentsGG},  marginalization with respect to $Z$ allows to derive the monents of the stationary search cost distribution when the request probabilities come from a two--parameter Poisson--Dirichlet process. 
\begin{thm}\label{MomentsGG}
If the $(p_1,\ldots,p_n)$ are determined by normalizing the increments of a generalized gamma subordinator with $t_i-t_{i-1}=1/n$ 
for each $i \in \{1, \ldots, n\}$, then
\begin{equation}\label{GGPSI}
\Psi(l)=\left\{ \begin{array}{ll}
\frac{(l!)^2}{((1/\gamma)-l-1)_{l}}\sum_{m=0}^l \frac{u^{m\gamma}}{m!\gamma^m}&\mbox{ if }\gamma<\frac{1}{k+1}\\ \infty &{\mbox{otherwise}}\\
\end{array}\right.
\end{equation}
If $(p_1,\ldots,p_n)$ are determined by the increments $w_i$'s in equation \eqref{2Pincre} then
\begin{equation}\label{PYPSI}
\Psi(l)=\left\{ \begin{array}{ll}
l!\frac{(\theta/\gamma+1)_l}{(\frac{1}{\gamma}-l-1)_{l}}&\mbox{ if }\gamma<\frac{1}{k+1}\\ \infty &{\mbox{otherwise}}\\
\end{array}\right.
\end{equation}
\end{thm}
\begin{proof} First of all, note that, in both cases, it is immediate to check that $\sup_{n} I_n(l)<\infty$, for every $l=1,\dots,k$, with $k$ such that $\gamma<\frac{1}{k+1}$. From equation \eqref{Momenti}, it is easy to see that   
\begin{multline*}
\Psi(l)= - 
    \frac{\gamma-1}{\gamma^l}\,\edr^{u^\gamma/\gamma}\\
 \times\,
    \int_0^{+\infty} \int_0^{+\infty}
    [(u+x+y)^{\gamma}-(u+x)^{\gamma}]^l
    (u+x+y)^{\gamma-2}\edr^{-(u+x)^{\gamma}/\gamma}\: \ddr x\,\ddr y.
  \end{multline*}
A simple change of variable along with the integrability condition $\gamma<\frac{1}{k+1}$ lead to 
\[
\Psi(l)=-(1-\gamma)\,\edr^{u^{\gamma}/\gamma}\:
\Gamma(l+1,u^{\gamma}/\gamma)\:
\sum_{r=0}^l \binom{l}{r} (-1)^{l-r}\frac{1}{(r+1)\gamma-1}\\
\]
where $\Gamma(a,x)=\int_x^{\infty} z^{a-1}e^{-z}dz$ is the incomplete gamma function, for any $x>0$. 
From 0.160.2 in \cite{gr} it follows that
\begin{align*}
\Psi(l)
&=\frac{1-\gamma}{\gamma}\:\edr^{u^{\gamma}/\gamma}
\Gamma(l+1,u^{\gamma}/\gamma)B(l+1,(1/\gamma)-l-1)\\[4pt]
&=\frac{l!}{((1/\gamma)-l-1)_{l+1}}\:\edr^{u^{\gamma}/\gamma}\:
\Gamma(l+1,u^{\gamma}/\gamma)
\end{align*}
Since $l\in\mathbb{N}$, then $\Gamma(l+1,u^{\gamma}/\gamma)$ is a polynomial in $(u^{\gamma}/\gamma)$, i.e.
\begin{equation*}\label{incompleta}
\Gamma(l+1,u^{\gamma}/\gamma)=e^{-u^{\gamma}/\gamma}\sum_{m=0}^l \frac{l!}{m!}\:\Big(\frac{u^{\gamma}}{\gamma}\Big)^m
\end{equation*}
and this concludes the proof of equation \eqref{GGPSI}. Equation \eqref{PYPSI} is proved by using the characterization of the two-parameter Poisson-Dirichlet process through mixture of generalized Gamma processes. It is straightforward to see that 
\begin{equation}\label{condizionale}
\lim_{n\rightarrow\infty} \mathbb{E}[S_n^k|Z=z]=
\sum_{l=1}^k  a_l^{k}\frac{(l!)^2}{(\frac{1}{\gamma}-l-1)_{l}}\sum_{m=0}^l \frac{z^{m}}{m!}
\end{equation}
Integrating over $z$, and noting that
$$\sum_{m=0}^l \frac{(\theta / \gamma)_m}{m!}=\frac{(\theta/\gamma+1)_l}{l!}$$
provides the desired result.
\end{proof}
Just to give an idea of the behaviour of the role of the parameters $(\theta,\gamma)$ in determining the search cost distribution, we plot in Figure 1 the first and the second moment of the stationary search cost.

\medskip

\noindent {\bf Remark: }It is possible to recover the result in \cite{llp} displayed in Theorem 1. Indeed, $u=0$ in equation \eqref{GGPSI} leads to their result. On the other hand, when $\gamma$ goes to zero, it is easy to see, in equation \eqref{PYPSI}, that
$$\lim_{n\rightarrow +\infty}\mathbb{E}[S_n^k]=\sum_{l=1}^k a_l^{(k)}\frac{\theta^l}{l!}.$$
From the equation above, one can recover the result for $k=1$ provided in \cite{kingman}. 
\begin{figure}[H]\label{Figura}
\centering
\begin{subfigure}{0.48\textwidth}
\centering
	 \raisebox{-\height}{\includegraphics[width=0.9\textwidth]{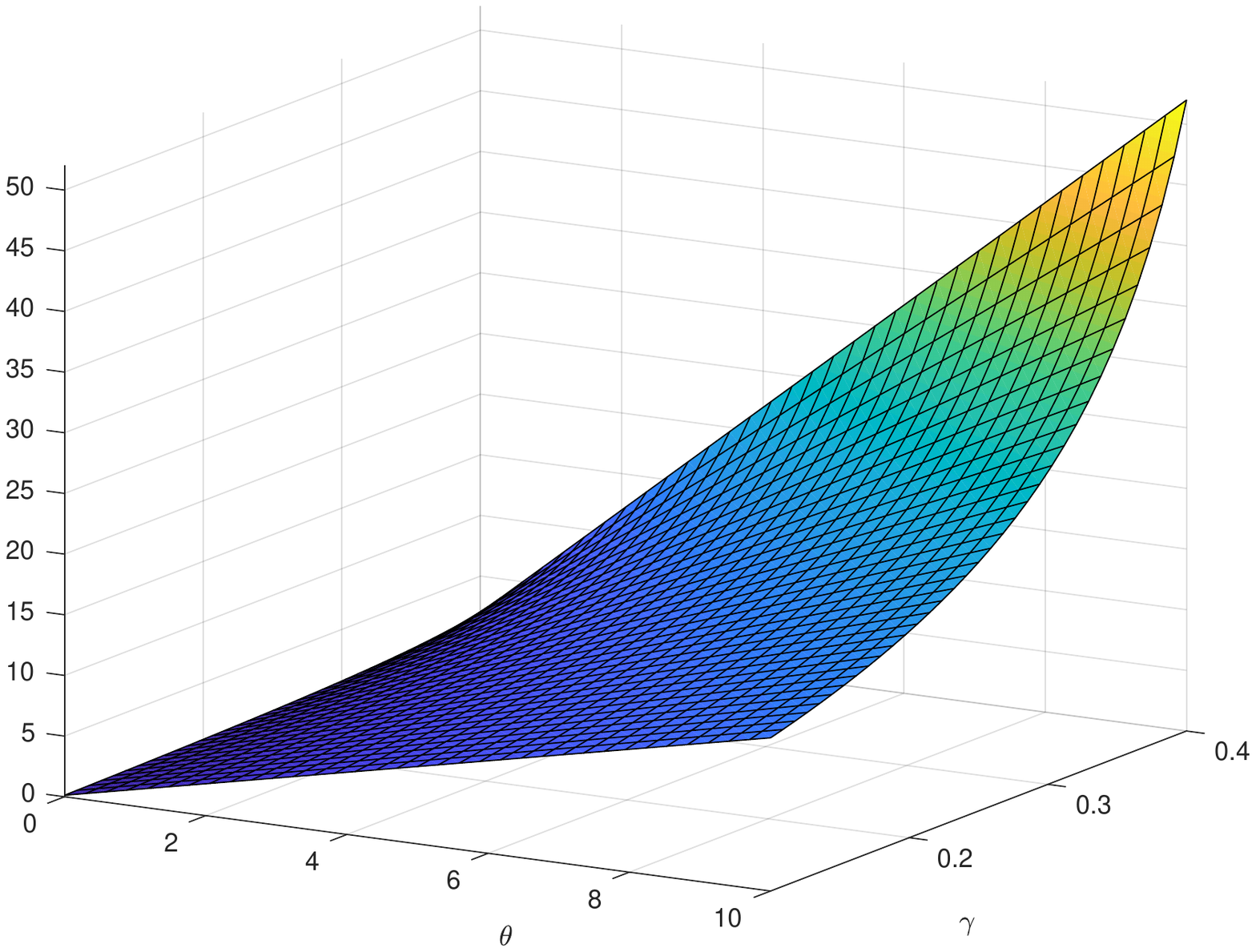} }
\end{subfigure}
\hspace{1em}
\begin{subfigure}{0.48\textwidth}
\centering
	\raisebox{-\height}{ \includegraphics[width=0.9\textwidth]{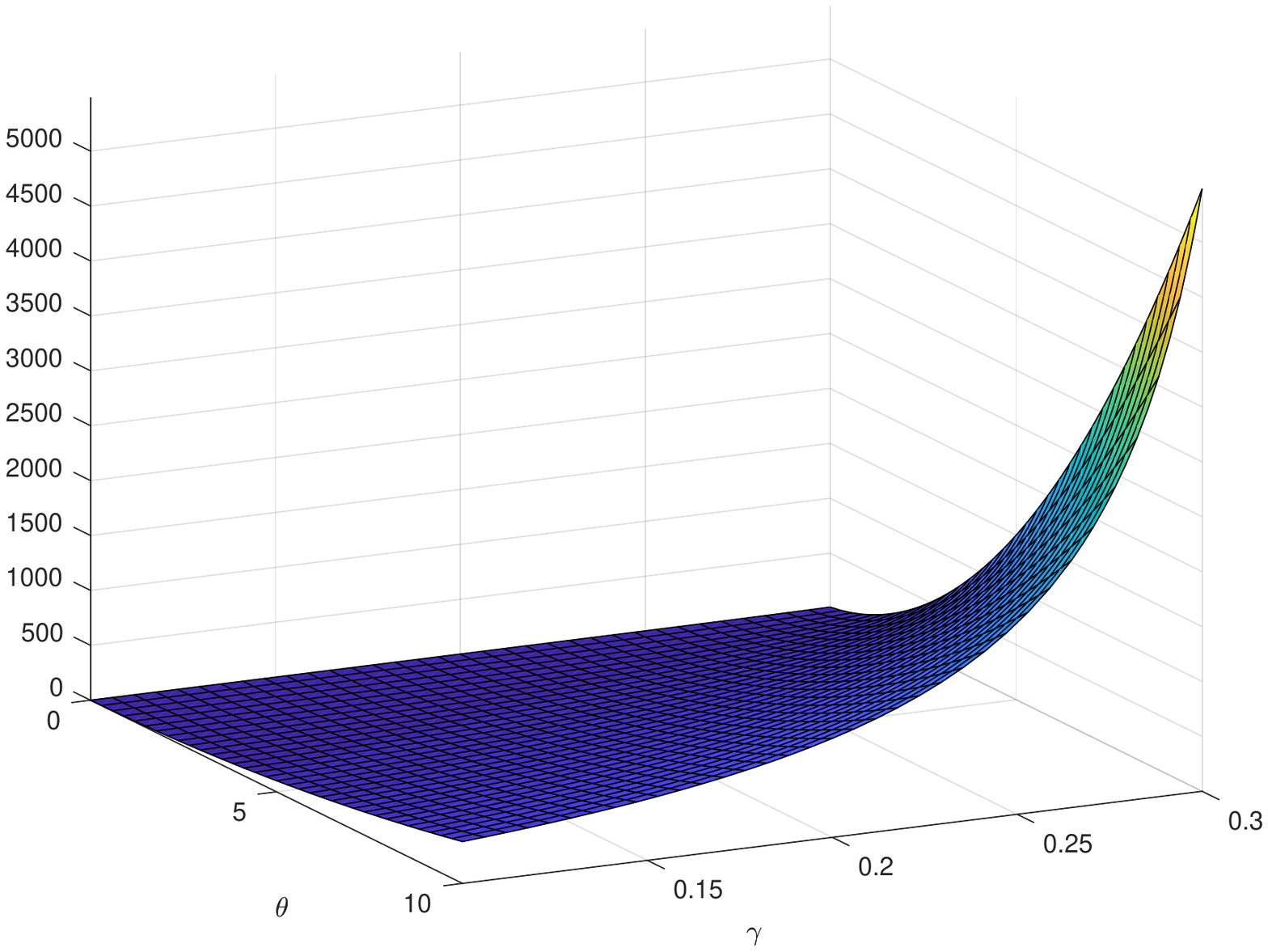} }
\end{subfigure}
\caption{The left figure shows the limiting behaviour of the first moment when $0<\theta<10$ and $0.1<\gamma<0.4$. The right figure shows the limiting behaviour of the first moment when $0<\theta<10$ and $0.1<\gamma<0.3$.}
\end{figure}

%

\end{document}